\documentclass[%
 reprint,
groupedaddress,
 amsmath,amssymb,
 aps,
 prx,
 floatfix,
]{revtex4-2}

\usepackage{silence}
\usepackage[bookmarksnumbered,linktocpage,hypertexnames=false,colorlinks=true,linkcolor=blue,urlcolor=blue,citecolor=blue,anchorcolor=green,breaklinks=true,pdfusetitle]{hyperref}
\usepackage{bm,bbm,bbold,braket,amsthm,amsfonts,mathtools,graphicx,enumitem,float,mathdots,enumitem,mleftright,framed,overpic,booktabs}

\usepackage[dvipsnames]{xcolor}
\usepackage{tikz}
\usepackage{quantikz}
\usepackage{dsfont}
\usepackage{soul} % strike through with \st{}
\usepackage{printlen}
\usepackage{chemformula}

\usepackage[longend,ruled,linesnumbered,nokwfunc]{algorithm2e}
\SetKwInput{Input}{Input}
\SetKwInput{Output}{Output}
\SetKwFor{Loop}{loop}{}{end loop}

\usepackage[capitalize,nameinlink]{cleveref}

\tikzstyle{box} = [rectangle, rounded corners, 
minimum width=3cm, 
minimum height=1cm,
text centered]
\tikzstyle{arrow}=[->]

\newtheorem{thm}{Theorem}[section]\crefname{thm}{Theorem}{Theorems}
\newtheorem*{thm*}{Theorem}
\newtheorem{lem}[thm]{Lemma}\crefname{lem}{Lemma}{Lemmas}
\crefname{prop}{Proposition}{Propositions}
\crefname{cor}{Corollary}{Corollaries}
\theoremstyle{definition}
\crefname{def}{Definition}{Definitions}
\theoremstyle{remark}

\newcommand{\del}{\partial}

\begin{document}

\title{Relative phase and dynamical phase sensing in a Hamiltonian model of the optical SU(1,1) interferometer}

%\date{\today}
\author{T.J.\,Volkoff}
\affiliation{Theoretical Division, Los Alamos National Laboratory, Los Alamos, NM, USA}

\begin{abstract}
The $SU(1,1)$ interferometer introduced by Yurke, McCall, Klauder \cite{PhysRevA.33.4033} is reformulated starting from the Hamiltonian of two identical optical downconversion processes with opposite pump phases. From the four optical modes, two are singled out up to a relative phase by the assumption of exact alignment of the interferometer (i.e., mode indistinguishability). The state of the two resulting modes is parametrized by the nonlinearity $g$, the relative phase $\phi$, and a dynamical phase $\theta$ resulting from the interaction time. The optimal operating point for sensing the relative phase (dynamical phase) is found to be $\phi = \pi$ ($\theta=0$) with quantum Fisher information exhibiting Heisenberg scaling $E^{2}$ (logarithmically modified Heisenberg scaling $\left({E\over \ln E}\right)^{2}$). Compared to the predictions of the circuit-based model, we find in that in the Hamiltonian model: 1. the optimal operating points occur for a non-vacuum state inside the interferometer, and 2. measurement of the total photon number operator does not provide an estimate of the relative or dynamical phase with precision that saturates the quantum Cram\'{e}r-Rao bound, whereas an observable based on weighted shift operators becomes optimal as $g$ increases. The results indicate a first-principles approach for describing general optical quantum sensors containing multiple optical downconversion processes.
\end{abstract}

\maketitle

\section{Introduction}
The $SU(1,1)$ interferometer was introduced by Yurke, McCall, and Klauder by replacing the $SU(2)$ beamsplitter transformations of a quantum circuit-based description of the Mach-Zehnder interferometer with corresponding $SU(1,1)$ transformations \cite{PhysRevA.33.4033}. The construction led to an astounding prediction: whereas a port of an $SU(2)$ interferometer must be seeded with non-classical light in order to obtain an estimate of the optical phase with mean squared error scaling as $N^{-2}$ (with $N$ being the expected photon number per mode), the ports of an $SU(1,1)$ interferometer can be left in vacuum while obtaining the same scaling of the precision. This difference in the input optical resource requirements is due to the fact that $SU(1,1)$ transformations describe \textit{active} optical dynamics, in particular, the optical sector of the Hamiltonian of a coherently pumped light-matter system containing a two-photon downconversion process \cite{mandelwolf}. Optically active dynamics produces a non-classical Gaussian state of light inside the $SU(1,1)$ interferometer, even from vacuum input. The description of the $SU(1,1)$ interferometer as a composition of optically active dynamics, in which entangled, phase shifted light from a first optical downconversion process seeds a temporally subsequent optical downconversion with a different phase \cite{PhysRevA.33.4033}, has inspired a general mechanism by which non-linear interactions can be used for quantum sensing protocols starting from unentangled quantum states, with light-based implementations \cite{PhysRevLett.119.223604,PhysRevA.96.053863,10.1063/5.0207332,Tian:23}, and matter-based implementations \cite{PhysRevLett.127.183401,Mao2023,PhysRevLett.116.053601,yao,PhysRevA.110.062425,PhysRevA.108.043302}, including trapped ions \cite{wineland}. The similarities between the quantum circuit description of optical interferometers and the quantum circuit description of quantum computation has also led to interferometers being considered as building blocks for quantum algorithms \cite{Knill2001,mosca,nielsen}.

Notwithstanding its enormous influence, the framework of \cite{PhysRevA.33.4033} leaves many open questions which have been instructively raised or revisited in the literature. For instance, the signal-to-noise ratio of the total photon number readout requires obtains Heisenberg scaling only at second order in the phase parameter. This fact partly motivated  development of a framework in which the $SU(1,1)$ interferometer is considered as a device suited for sensing relative phase space displacements between optical modes \cite{cavesreframing}. Analyses of noisy $SU(1,1)$ interferometers have led to the understanding that downconversion process is more correctly viewed as a composition of noisy quantum channels instead of a unitary two-mode squeezer \cite{PhysRevA.78.043816}. In the present work, which like \cite{PhysRevA.33.4033} is motivated by developing a workable analogy with the $SU(2)$ case, we develop a Hamiltonian description of the $SU(1,1)$ interferometer and analyze its optimal sensitivity with respect to the phase parameters of the system. The results are compared with quantum circuit-based approaches \cite{PhysRevA.33.4033,PhysRevA.85.023815}.

Specifically, our analysis of the maximal quantum Fisher information (QFI) and optimal readout for the relative phase $\phi$ of the Hamiltonian model in Section \ref{sec:relphase} shows that Heisenberg scaling $O(E^{2})$ of the sensitivity with respect to the total energy $E$ is obtained at the operating point $\phi=\pi$. Since the relative phase is not a native parameter of circuit-based models, this provides a prediction of Heisenberg scaling for a kinematic phase parameter of an $SU(1,1)$ interferometer. For the task of estimation of the dynamical phase $\theta$, the maximal QFI is found to occur at the operating point $\theta=0$, like in the circuit-based model of \cite{PhysRevA.33.4033}, but this operating point corresponds to non-trivial dynamics in the Hamiltonian-based model of the present work, whereas the corresponding quantum circuit is the identity operator at this point. For both phase parameters, we describe an optimal readout exhibiting a near-cancellation in noise contributions from total intensity and two-photon coherence. These sensitivity results, along with basic differences in optical properties such as intensity at various operating points, provide experimental targets for falsification of either model. Less dramatically, the results provide further evidence that different operating modes of the device of \cite{PhysRevA.33.4033} deserve different dynamical descriptions, a point that is also motivated by analyses of noisy $SU(1,1)$ interferometers  \cite{PhysRevA.78.043816,PhysRevA.86.023844}.

A general conceptual reason to exercise caution when applying the quantum circuit model to optical systems is provided by the fact that while a Hamiltonian quadratic in the boson creation and annihilation operators of $m$ optical modes (i.e., an element of the symplectic Lie algebra $\mathfrak{sp}(2m,\mathbb{R})$) always generates a time-evolution corresponding to an element of $Sp(2m,\mathbb{R})$ \cite{holevo, folland}, the logarithm of a product of two or more $Sp(2m,\mathbb{R})$ group elements is not necessarily an element of $\mathfrak{sp}(2m,\mathbb{R})$. Quantum circuit-based descriptions of active optical networks therefore lead to predictions about dynamics that are not predicted by Hamiltonian-based descriptions \cite{PhysRevResearch.6.L042011}. Section \ref{sec:sss1} details the construction of both of these models of the $SU(1,1)$ interferometer, and in Section \ref{sec:dynphase} we show that the models lead to conflicting results for the sensing performance of the $SU(1,1)$ interferometer and the structure of the quantum state at the optimal operating point. In the next section, we show that these conflicts are absent in passive optical interferometers.

\section{Warm up: no difference between a circuit-based and Hamiltonian-based $SU(2)$ interferometer}

The Hamiltonian describing the interior of a two-mode Mach-Zehnder interferometer at some symmetrically chosen interaction time is that of the free electromagnetic field on two orthogonal modes $\ket{1}$ and $\ket{2}$,
\begin{equation}
H={\theta \over 2}b_{1}^{\dagger}b_{1} - {\theta \over 2}b_{2}^{\dagger}b_{2}.
\end{equation}
Different choices of the interaction time add a term proportional to the constant of motion $\sum_{j=1}^{2}b_{j}^{\dagger}b_{j}$. However, one desires that the interferometer be sensitive to changes in $\theta$ even only one beam of light is available, e.g., one has access to coherent light in a mode $\ket{1'}$ and vacuum in another mode $\ket{2'}$. Properly aligned, lossless, symmetric beamsplitters allow one to make the modes $\ket{1}$, $\ket{2}$ indistinguishable from hybridized modes ${\ket{1'}+e^{i\phi}\ket{2'}\over\sqrt{2}}$, ${-e^{-i\phi}\ket{1'}+\ket{2'}\over\sqrt{2}}$. This completes the Mach-Zehnder device, which is described by the Hamiltonian
\begin{align}
H'=& {\theta \over 2}\left[ \left( {a_{1'}^{\dagger}+e^{i\phi}a_{2'}^{\dagger}\over \sqrt{2}} \right)\left( {a_{1'}+e^{-i\phi}a_{2'}\over \sqrt{2}}  \right) \right. \nonumber \\
&{} \left. - \left( {-e^{-i\phi}a_{1'}^{\dagger}+a_{2'}^{\dagger}\over \sqrt{2}} \right)\left( {-e^{i\phi}a_{1'}+a_{2'}\over \sqrt{2}}  \right) \right] .
\end{align}
The parameters $\theta$ and $\phi$ have different physical interpretations: $\phi$ is a relative phase in $[0,2\pi)$ whereas for continuous-variable quantum systems without a particle number superselection rule, $\theta$ is not a phase in the sense of being unitless.  Specifically, the parameter $\theta$ is an interaction time, whereas $\phi$ is a phase shift from, e.g., optical path length changes or beamsplitter dielectrics. We call $\theta$ a dynamical phase because its product with $\hbar^{-1}$ results in a unitless phase ($\hbar=1$ in this work).
It is then clear from the group law of $SU(2)$ that for any input state $\ket{\eta}_{1'}\ket{\xi}_{2'}$,
\begin{align}
e^{-iH'}\ket{\eta}_{1'}\ket{\xi}_{2'} = &e^{i{\pi\over 2}(\sin \phi J_{x} - \cos \phi J_{y})}\nonumber \\
&{} \cdot e^{-i\theta J_{z}}\nonumber \\
&{} \cdot e^{-i{\pi\over 2}(\sin \phi J_{x} - \cos \phi J_{y})}\ket{\eta}_{1'}\ket{\xi}_{2'}
\label{eqn:su2}
\end{align}
where the quantum circuit on the right hand side involves the Schwinger boson operators $J_{x}$, $J_{y}$, $J_{z}$ with respect to the $\ket{1'}$ and $\ket{2'}$ modes. Note that the state cannot be used to probe the parameter $\phi$ if $\theta = 0$. The Hamiltonian model of the interferometer allows one to keep track of the interactions between the modes, however simple or complicated they may be, while the circuit-based model allows a phase space interpretation of the Mach-Zehnder device. Since the models agree on the quantum state, discussion of the $SU(2)$ interferometer with either model is a matter of preference. We emphasize that the difference between the models is not the distinction between the Schr\"{o}dinger and Heisenberg pictures of the quantum dynamics. No observable or readout has been specified; both models describe the Schr\"{o}dinger picture dynamics of the input state.

\section{Hamiltonian and circuit-based models of two-mode $SU(1,1)$ interferometer\label{sec:sss1}}

An optical diagram for the two-mode $SU(1,1)$ interferometer is given in Fig. 6 of \cite{PhysRevA.33.4033}, and in Fig. \ref{fig:ooo} with the mode labeling and considered in the present work. Fig. \ref{fig:ooo} shows two optical downconversion processes that are pumped by a coherent beam that incurs a $\pi$ phase shift between the processes. For consistency, we refer to these processes by  FWM1 and FWM2 (with FWM standing for four-wave mixing) and, initially, consider them to be independent processes. Specifically, without any assumption on the relationship between the modes, the Hamiltonian in the undepleted pump approximation is obtained by summing the contributions from FWM1 and FWM2
\begin{align}
H&= H_{0}+ H_{\text{int}} \nonumber \\
H_{\text{int}}&:= g\left( a_{1}^{\dagger}a_{2}^{\dagger} +h.c.\right) - g\left( b_{1}^{\dagger}b_{2}^{\dagger} +h.c.\right)
\label{eqn:hhh}
\end{align}
where the downconversion modes from FWM1 are associated with $a_{1}$, $a_{2}$ and those from FWM2 are associated with $b_{1}$, $b_{2}$, respectively, and the Hamiltonian $H_{0}$ is the free field Hamiltonian for a set of orthogonal modes in the Hilbert space spanned by $a_{1}^{\dagger}\ket{0}$, $a_{2}^{\dagger}\ket{0}$, $b_{1}^{\dagger}\ket{0}$, $b_{2}^{\dagger}\ket{0}$. The parameter $g>0$ takes into account the pump strength and the $\chi^{(2)}$ nonlinearity of the material.

We now relate the output modes of FWM1, which we call $a_{1}$, $a_{2}$, to the output modes of FWM2, which we call $b_{1}$, $b_{2}$, by complex scalars that are ultimately determined by the optical elements in the interior of the interferometer and their alignment. Mathematically, the pairs of modes $a_{1}$ and $b_{1}$ and $a_{2}$ and $b_{2}$ are partially distinguishable, in a way that is specified by the commutation relations
\begin{align}
[a_{i},a_{j}^{\dagger}]&=[b_{i},b_{j}^{\dagger}]=\delta_{ij}\mathbb{I} \nonumber \\
[a_{1},b_{2}^{\dagger} ]&= z_{1}\mathbb{I} \; , \;  \vert z_{1}\vert\in [0,1]\nonumber \\
[a_{2},b_{1}^{\dagger} ]&= z_{2}\mathbb{I} \; , \;  \vert z_{2}\vert\in [0,1].
\label{eqn:comm}
\end{align}
If $\vert z_{i}\vert <1$, unitarity requires the introduction of auxiliary environment modes to obtain the full description of the system \cite{PhysRevA.109.023704}. However, in the present work, we consider the modes $a_{1}$ and $b_{2}$ to be indistinguishable up to a phase shift in the wavefunction, and the modes $a_{2}$ and $b_{1}$ to be strictly indistinguishable. Therefore, the phase parameter $\phi$ is introduced by taking $z_{1}=e^{-i\phi}$, whereas the complex scalar $z_{2}=1$. In an experiment, one empirically determines the appropriate arguments of $z_{1}$ and $z_{2}$ according to the device, including a choice of reference phase.  The phase $\phi$ is analogous to the phase $\phi$ in the $SU(2)$ interferometer (\ref{eqn:su2}). Our decision to analyze  $\vert z_{i}\vert=1$ for the distinguishability parameters is motivated by the  alignment of the FWM1 and FWM2 processes proposed in \cite{PhysRevA.33.4033}, which is achieved by the use of two mirrors. This alignment analysis involves the same consideration of path indistinguishability that is central to the description of the Wang-Zou-Mandel effect \cite{PhysRevA.44.4614,PhysRevA.41.1597,PhysRevA.109.023704}. 

 \begin{figure}[t]
    \centering
    \includegraphics[scale=0.7]{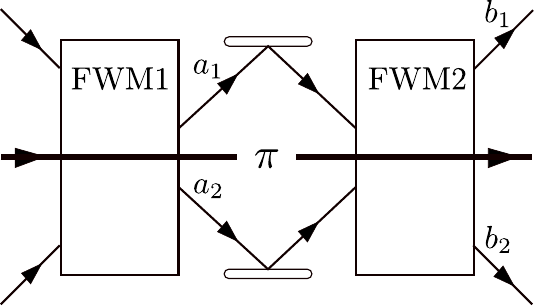}
    \caption{$SU(1,1)$ interferometer with downconversion modes labeled.}
    \label{fig:ooo}
\end{figure}

 Having taken the modes $a_{1}$, $b_{2}$ to be related by a phase $\phi$, and the modes $a_{2}$, $b_{1}$ to be identical, the Hamiltonian (\ref{eqn:hhh}) simplifies to a two-mode system
\begin{align}
H&=H_{0} + H_{\text{int}} \nonumber \\
H_{0}&= \theta(a_{1}^{\dagger}a_{1}+a_{2}^{\dagger}a_{2})\nonumber \\
H_{\text{int}}&= {2g\sin{\phi\over 2}}\left( ie^{-i{\phi\over 2}}a_{1}^{\dagger}a_{2}^{\dagger} +h.c.\right)
\label{eqn:hhhh}
\end{align}
which is obtained simply by taking $b_{2}^{*}=e^{-i\phi}a_{1}^{*}$ and $b_{1}^{*}=a_{2}^{*}$ in (\ref{eqn:hhh}).  Because the difference $a_{1}^{\dagger}a_{1}-a_{2}^{\dagger}a_{2}$ is an $SU(1,1)$ constant of the motion, only the total photon number operator $a_{1}^{\dagger}a_{1}+a_{2}^{\dagger}a_{2}$ appears in the free field Hamiltonian. Therefore, the free field Hamiltonian is specified in a way that is exactly analogous to (\ref{eqn:su2}), with the parameter $\theta$ again having the interpretation of interaction time. The probe state is given by $e^{-iH}\ket{0}_{1}\ket{0}_{2}$.   This parameterized state merits the name ``$SU(1,1)$ interferometer'' due to the fact that $H$ is an element of an $\mathfrak{sp}(2,\mathbb{R})$ subalgebra of $\mathfrak{sp}(4,\mathbb{R})$ in its two-mode bosonic form. Note that the generators of $\mathfrak{sp}(2,\mathbb{R})$
\begin{align}
    K_{1}&= {a_{1}^{\dagger}a_{2}^{\dagger} +h.c. \over 2} \nonumber \\
    K_{2}&= {-ia_{1}^{\dagger}a_{2}^{\dagger} +h.c. \over 2} \nonumber \\
    K_{3}&= {a_{1}^{\dagger}a_{1}+ a_{2}^{\dagger}a_{2} +1 \over 2}
    \label{eqn:br}
\end{align}
 can be converted to generators of $\mathfrak{su}(1,1)$ by multiplying by $i$, so the Lie algebras are isomorphic \footnote{The global relationship between the real symplectic groups and the pseudounitary group can be stated by the  following strict containments as Lie subgroups: $Sp(2n,\mathbb{R})\subset U(n,n)\subset Sp(4n,\mathbb{R})$.}. We emphasize that $H$ is the Hamiltonian of a system of two downconversion processes that only takes into account a relative $\pi$ phase shift of the pump between the processes, not the global phase of the pump. In particular, the phase $\phi$ is a property of the wavefunctions $a_{1}$ and $b_{1}$, not a property of the pump mode, and this is why the effective nonlinearity $g$ is renormalized by a $\phi$-dependent factor $2\sin{\phi\over 2}$. The Hamiltonian for a general downconversion process in the undepleted pump approximation can be written as an element of $\mathfrak{sp}(2,\mathbb{R})$ with a coupling constant $ge^{2i\Phi}$ where $g$ is the bare nonlinearity proportional to the $\chi^{(2)}$ susceptibility of the nonlinear optical medium and $\Phi$ is the global pump phase \cite{PhysRevA.39.2493,wm} that we neglect in the present analysis; this was also neglected by Yurke, McCall, Klauder. We also note that $\mathfrak{sp}(2,\mathbb{R})$ also has a single-mode bosonic realization, so a Hamiltonian model of a $SU(1,1)$ interferometer can also be developed for a single optical mode by first considering two independent degenerate downconversion processes $g(a^{\dagger 2}+a^{2})$ and $-g(b^{\dagger 2}+b^{2})$, then applying the analogous mode distinguishability considerations above. In Appendix \ref{sec:aaa} we show the covariance matrix for the full state $e^{-iH}\ket{0}\ket{0}$ and its symplectic eigenvalues, which determine its squeezing and entanglement properties.

To contrast with the resulting Hamiltonian model $e^{-iH}$, consider the dynamics defined by the quantum circuit \cite{Plick_2010} 
\begin{equation}
V = e^{-g(a_{1}^{\dagger}a_{2}^{\dagger}-h.c.)} e^{i\theta a_{2}^{\dagger}a_{2}} e^{g(a_{1}^{\dagger}a_{2}^{\dagger}-h.c.)}.
\label{eqn:ccc}
\end{equation}
The relative phase $\phi$ is absent in this model, and it is standard in the literature to consider a single phase parameter which couples to the photon number operator. This is natural because the first and last circuit elements are understood to describe sequential downconversion processes, and replacing the generator $a_{1}^{\dagger}a_{2}^{\dagger}$ with $e^{-i\phi}a_{1}^{\dagger}a_{2}^{\dagger}$ would give $\phi$ the same form as a global pump phase $2\Phi$ discussed above. Further, the $\mathfrak{su}(1,1)$ generators in the first and last circuit layers couple to the bare nonlinearity $g$ so there is no notion of effective (i.e., renormalized) nonlinearity.
The unitary two-mode squeezing operators at the beginning and end of the circuit (\ref{eqn:ccc}) are elements of $SU(1,1)$.
However, the middle phase shift is not an element of $SU(1,1)$. This can be remedied by noting that the circuit is invariant under insertion of a circuit element $e^{i\theta '(a_{1}^{\dagger}a_{1}-a_{2}^{\dagger}a_{2})}$ anywhere in the product, and this fact can be used to convert (\ref{eqn:ccc}) to a composition of $SU(1,1)$ operations.  Specifically, taking $\theta ' = \theta/2$ gives 
\begin{equation}
V = e^{-g(a_{1}^{\dagger}a_{2}^{\dagger}-h.c.)} e^{i{\theta \over 2} \left( a_{2}^{\dagger}a_{2} + a_{1}^{\dagger}a_{1} \right)} e^{g(a_{1}^{\dagger}a_{2}^{\dagger}-h.c.)}
\label{eqn:su11ccc}
\end{equation}
which has its middle phase shift in $SU(1,1)$.
Similarly, the composition of scattering matrices in the original formulation of the $SU(1,1)$ interferometer in \cite{PhysRevA.33.4033}, i.e., the Heisenberg picture transformation of the modes $a_{1}$ and $a_{2}$, is implemented by the circuit
\begin{equation}
V=e^{-i{g\over 2}\left(a_{1}^{*}a_{2}^{*}+h.c.\right)}e^{i\left( \theta_{1}a_{1}^{\dagger}a_{1} + \theta_{2}a_{2}^{\dagger}a_{2} \right)}e^{i{g\over 2}\left(a_{1}^{*}a_{2}^{*}+h.c.\right)}.
\label{eqn:ykyk}
\end{equation}
The middle circuit element can be converted to an $SU(1,1)$ operation using the same trick, with the resulting circuit being a function of $g$ and  $\theta=\theta_{1}+\theta_{2}$.

A similar circuit-based model involving two downconversion processes appears in \cite{10.1063/1.3606549,PhysRevA.85.023815,Hudelist2014, 	jing}.  This alternative model does not utilize an inverted downconversion process for the last circuit element, but is actually equal to (\ref{eqn:ccc}) when the latter is seeded with a two-mode squeezed vacuum produced by a downconversion with nonlinearity $2g$. We therefore restrict our analysis of circuit-based models to the form (\ref{eqn:ccc}) or, equivalently, to $SU(1,1)$ elements of the form (\ref{eqn:su11ccc}).

 If the circuit $V$ is seeded by non-vacuum states, the structure of the phase in (\ref{eqn:su11ccc}) is crucial because the QFI depends on how the dynamical phase $\theta$ is translated into the quantum circuit model \cite{Gong_2017}. In this case, the discrepancies can be resolved by calculating the QFI for the various phase shift circuits with a dephased probe state that models the absence of an external phase reference \cite{PhysRevA.99.042122}. Because our present goal is to compare the circuit-based model with the Hamiltonian-based model only in terms of dynamics, we will restrict both models to vacuum inputs.

As discussed above, the Hamiltonian model (\ref{eqn:hhhh}) is derived from the Hamiltonian of the quantum electromagnetic field of two independent pairs of modes simply by assuming indistinguishability of modes up to relative phase shifts. This direct relationship to the quantum optical dynamics motivates one to examine the Hamiltonian corresponding to the circuit-based dynamics in (\ref{eqn:su11ccc}). However, in doing this, one runs into the fact that the composition of $SU(1,1)$ operators in (\ref{eqn:su11ccc}) is a fine-tuned description of the dynamics. Specifically, although there is a Hamiltonian in $\mathfrak{sp}(2,\mathbb{R})$ that exponentiates to (\ref{eqn:su11ccc}) via $e^{-iH}$ (due to the fact that the first and last layer are exactly inverse), this is not generally true for the parametrized circuits
 \begin{equation}
V(g_{1},g_{2},\theta)= e^{g_{2}(a_{1}^{\dagger}a_{2}^{\dagger}-h.c.)} e^{i{\theta\over 2} \left( a_{2}^{\dagger}a_{2} + a_{1}^{\dagger}a_{1} \right)} e^{g_{1}(a_{1}^{\dagger}a_{2}^{\dagger}-h.c.)}.
\label{eqn:gencirc}
\end{equation}
In fact, for $g_{2} = -g_{1}\pm \epsilon$ and $g_{1}>0$, $0<\epsilon<g_{1}$, for which $V(-g\pm \epsilon,g,\theta)$ appears to describe an arbitrarily small mismatch in the magnitude of nonlinearity of the downconversion process in Fig. \ref{fig:ooo} (due to, e.g., intensity fluctuations of the pump laser), we show in (\ref{eqn:uhuh}) below that no such Hamiltonian in $\mathfrak{sp}(2,\mathbb{R})$ exists. Indeed, the exponential map $\text{expi}:\mathfrak{sp}(2,\mathbb{R})\rightarrow SU(1,1)$ given by $\text{expi}H := e^{-iH}$ is not surjective, nor has dense image \cite{Chiribella2006,TORRE2005111}, indicating that certain open subsets of $SU(1,1)$ CV circuits do not correspond to a Hamiltonian that is quadratic in the creation and annihilation operators \cite{PhysRevResearch.6.L042011}. On the other hand, a quadratic Hamiltonian is fundamental in the analysis of the quantum electromagnetic field of a system of optically pumped downconversion processes in the undepleted pump approximation, so the Hamiltonian describing the optical diagram in Fig.\ref{fig:ooo} is an element of $\mathfrak{sp}(2,\mathbb{R})$.  We therefore contend that fine-tuned CV circuits such as (\ref{eqn:su11ccc}) are not appropriate for describing Hamiltonian dynamics of the $SU(1,1)$ interferometer, at least with the parameters $g$ and $\theta$ having the same meaning as downconversion nonlinearity and dynamical phase, respectively. Finding a continuous-variable (CV) Gaussian circuit decomposition of the dynamics $e^{-iH}$ using a product $\prod_{\ell=1}^{L}e^{-iH_{\ell}}$ of time-evolutions generated by quadratic Hamiltonians is a problem of CV-to-CV quantum simulation or compiling a CV quantum circuit \cite{PRXQuantum.2.040327,PhysRevLett.126.190504,Chiribella2006}, which we consider an important target for future work for digital quantum simulation of nonlinear optical interferometers.
 
 To concretely illustrate the point in the case of the quantum dynamics of Fig. \ref{fig:ooo}, we consider the  circuit dynamics (\ref{eqn:gencirc}) which is an element of $SU(1,1)$.
 To determine the existence of $H\in \mathfrak{sp}(2,\mathbb{R})$ such that $e^{-iH}$ is equal to (\ref{eqn:gencirc}), consider first the defining representation $\mathbb{R}^{2}$ in which $Sp(2,\mathbb{R})$ is given by the  $2\times 2$ matrices that preserve the symplectic form $\Delta(z,w)=z_{1}w_{2}-z_{2}w_{1}$. The Lie algebra $\mathfrak{sp}(2,\mathbb{R})$ can be given by $K_{1}={-i\over 2}X$, $K_{2}={-i\over 2}Y$, $K_{3}={1\over 2}Z$ where $X$, $Y$, $Z$ are the Pauli matrices, and the $2\times 2$ matrix  generated by the Hamiltonian $H= \sum_{j=1}^{3}s_{j}K_{j}$ can be written $e^{-iH}\in SU(1,1)$. This matrix has trace greater than or equal to $-2$ \cite{Chiribella2006}. Using the relation between $\mathfrak{sp}(2,\mathbb{R})$ and $\mathfrak{su}(1,1)$, we write (\ref{eqn:gencirc}) as  $e^{g_{2}Y}e^{i{\theta \over 2}Z}e^{g_{1}Y}$. One verifies that for $\theta\in [\pi,3\pi]$ the trace condition for this matrix translates to the condition
 \begin{equation}
 \cosh (g_{1}+g_{2}) \le {1\over \vert \cos {\theta\over 2}\vert}.
 \label{eqn:uhuh}
 \end{equation}
 This condition is necessary and sufficient that there exists $H\in \mathfrak{sp}(2,\mathbb{R})$ such that $\log e^{g_{2}Y}e^{i{\phi\over 2}Z}e^{g_{1}Y} =iH$, i.e., that there exists a quadratic Hamiltonian  $H$ corresponding to the circuit (\ref{eqn:gencirc}). The fine tuned circuits $V(-g,g,\theta)$ constitute a measure zero submanifold of circuits $V(g_{1},g_{2},\theta)$ for which such an $H$ exists. In the reverse direction of finding effective CV quantum circuits for $e^{-iH}$, one can note that the Cartan involution $\Theta(H) = H^{\dagger}$ on $\mathfrak{sp}(2,\mathbb{R})$ gives a direct sum decomposition of $\mathfrak{sp}(2,\mathbb{R})$ into $\mathfrak{p}=\text{span}_{\mathbb{R}}\lbrace K_{1},K_{2}\rbrace$ and $\mathfrak{l}=\mathbb{R}K_{3}$, from which it follows that $\text{expi}H \in \text{expi}\mathfrak{l}\cdot \text{expi}\mathfrak{p}\cdot \text{expi}\mathfrak{l}$ (i.e., the $KAK$ decomposition \cite{knapp}).

\section{Estimation of $\phi$ in the Hamiltonian model\label{sec:relphase}}
We first consider the case of estimation of the relative phase shift $\phi$, taking the interaction time (i.e., dynamical phase) $\theta=0$. This value for $\theta$ is not physical, but allows one to isolate the $g$-dependence of the precision of the estimator of $\phi$.   Taking $\theta=0$ and defining 
 \begin{align}
 \ket{\psi(g,\phi)}&:= e^{-iH_{\text{int}}}\ket{0}_{1}\ket{0}_{2} \nonumber \\
 &{}= e^{2g\sin{\phi\over 2}\left( e^{-i{\phi\over 2}}a_{1}^{\dagger}a_{2}^{\dagger} - h.c.\right)}\ket{0}_{1}\ket{0}_{2},
 \label{eqn:thzero}
 \end{align}
  the QFI is then calculated according to
 \begin{align}
 \text{QFI}(\phi)&= -2\del_{\phi'}^{2} \big\vert  \langle \psi(g,\phi)\vert \psi(g,\phi')\rangle \big\vert^{2} \Bigg\vert_{\phi'=\phi} \nonumber \\
 &= {1\over 4}\sinh^{2}\left( 4g\sin^{2}{\phi\over 2} \right) +4g^{2}\cos^{2}{\phi\over 2}
 \label{eqn:hamqfi}
 \end{align}
 which depends non-trivially on $\phi$ because, as a state-valued function of $\phi$, $e^{-iH_{\text{int}}}\ket{0}_{1}\ket{0}_{2}$ does not fall under the shift model form of quantum estimation \cite{holevo}. Although the global maximum occurs for $\phi=\pi$, the values of the relative maxima of the QFI at $\phi=\pi$ and $\phi=0$ are only distinguished at $O(g^{4})$. The $4\times 4$ covariance matrix of the parametrized state is
 \begin{align}
 \Sigma_{\ket{\psi(g,\phi)}}&=  ={1\over 2}\begin{pmatrix}A&B\\
 B&A\end{pmatrix} \nonumber \\
 A&= \cosh\left( 4g\sin{\phi\over 2} \right) \mathbb{I}_{2} \nonumber \\
 B&= \sinh\left( 4g\sin{\phi\over 2} \right)\left( \cos{\phi\over 2}Z - \sin{\phi\over 2}X\right)
 \label{eqn:cov}
 \end{align}
 with Pauli $X$ and $Z$ matrices.  
It is readily verified for the Hamiltonian based model $\ket{\psi(g,\phi)}$ that the total photon number observable $a_{1}^{\dagger}a_{1}+a_{2}^{\dagger}a_{2}$ has a signal-to-noise ratio that satisfies $\text{SN}_{O}(\phi=0)=\text{QFI}(\phi=0)$. However, $\phi=0$ is not a valid operating point of the interferometer for the same reasons that $\theta=0$ is not a valid operating point of the circuit-based model. Further, both the QFI and the total energy scale as $O(g^{2})$ in a neighborhood of $\phi=0$, so that in this domain the maximal precision is the same as a classical optical sensor. Therefore, the main task is to identify an observable with a signal-to-noise ratio scaling with the square of the energy at the optimal working point $\phi=\pi$. The total photon number does not work for this purpose, having zero signal-to-noise ratio at $\phi=\pi$. Similarly, utilizing the covariance matrix (\ref{eqn:cov}), one finds that arbitrary self-adjoint operators quadratic in the canonical observables $q_{1}$, $q_{2}$, $p_{1}$, $p_{2}$ do not serve the desired purpose. Via calculation of the classical Fisher information, one further finds that a homodyne measurement does not contain enough information to saturate the QFI at $\phi=\pi$. 

An observable with the desired scaling of the signal-to-noise ratio is obtained by examining the structure of the symmetric logarthmic derivative (SLD), which at $\phi=\pi$ is given by
\begin{align}
L_{\pi}&=-ia_{1}^{\dagger}a_{1}\ket{\psi(g,\pi)}\bra{\psi(g,\pi)}+i\ket{\psi(g,\pi)}\bra{\psi(g,\pi)}a_{1}^{\dagger}a_{1} \nonumber \\
&=\left( -2g\ket{1}_{1}\ket{1}_{2}\bra{0}_{1}\bra{0}_{2} + h.c. \right) + o(g)
\end{align}where $\sigma a_{1}^{\dagger}a_{1} + (1-\sigma) a_{2}^{\dagger}a_{2}$, $\sigma \in [0,1]$, can be used instead of $a_{1}^{\dagger}a_{1}$. The small $g$ expansion of the SLD \cite{Liu_2023} motivates consideration of the following observable
\begin{equation}
O=\sum_{n=0}^{\infty}n\left( \ket{n+1}\ket{n+1}\bra{n}\bra{n} +h.c.\right)
\label{eqn:oop}
\end{equation}
and its signal-to-noise ratio at fixed $g$
 \begin{equation}
\text{SN}_{O}(\phi)={\left(\del_{\phi}\bra{\psi(g,\phi)} O \ket{\psi(g,\phi)} \right)^{2}\over \text{Var}_{\ket{\psi(g,\phi)}}O}.
\end{equation}
Note that the operator $O$ does not give any signal from the vacuum component from a probe state. This does not affect the readout for the probe states considered in the present work because the vacuum component is independent of $\phi$.

One finds that $\langle \psi(g,\pi)\vert O \vert \psi(g,\pi)\rangle =0$ and that the signal is given by
\begin{align}
&{} \del_{\phi}\langle \psi(g,\phi)\vert O \vert \psi(g,\phi)\rangle= \nonumber \\
&{} \del_{\phi} \left[ {2\sum_{n=0}^{\infty}n\tanh^{2n+1}\left( 2g\sin{\phi\over 2}\right) \cos{\phi\over 2} \over \cosh^{2}\left( 2g\sin{\phi\over 2} \right)} \right]\nonumber \\
&= -\sinh^{2}(2g)\tanh(2g) \text{ at } \phi=\pi.
\end{align}
The noise is given by $\langle \psi(g,\pi)\vert O^{2} \vert \psi(g,\pi)\rangle$. To understand why this expression is $O(1)$ for all $g$, it is useful to expand the square
\begin{align}
O^{2}&=\sum_{n=1}^{\infty}n^{2}\left[ \vphantom{\sum} \ket{n+1}\ket{n+1}\bra{n+1}\bra{n+1} + \ket{n}\ket{n}\bra{n}\bra{n}\vphantom{\sum}\right] \nonumber \\
&{} + \sum_{n=1}^{\infty} \left[ \vphantom{\sum} n(n+1)\ket{n}\ket{n}\bra{n+2}\bra{n+2}  \right. \nonumber \\
&{} \left. + n(n-1)\ket{n+1}\ket{n+1}\bra{n-1}\bra{n-1}\vphantom{\sum}\right].
\label{eqn:osq}
\end{align}
Therefore, a given probe state has two sources of noise in a measurement of $O$: one from the total intensity (first line of (\ref{eqn:osq})) and one from two-photon coherences (second line of (\ref{eqn:osq})). However, for the probe state $\ket{\psi(g,\pi)}$, these come with opposite signs and nearly equal magnitudes. With $p := 1/\cosh ^{2}2g$ defining the geometric distribution $\lbrace p(1-p)^{k}\rbrace_{k=0}^{\infty}$ and $m_{j}$ the $j$-th moment of that distribution, one obtains
\begin{align}
 \langle \psi(g,\pi)\vert O^{2} \vert \psi(g,\pi)\rangle &= p^{2}m_{2} - 2p(1-p)m_{1}+(1-p)^{2}\nonumber \\
&= 1-p = \tanh^{2}2g.
\end{align}
Using (\ref{eqn:hamqfi}), one then obtains that ${\text{SN}_{O}(\pi)\over \text{QFI}(\pi) } \sim 1$ asymptotically with increasing $g$, indicating the asymptotic optimality of the observable $O$ for estimation of relative phase $\phi$.

\section{Estimation of the dynamical phase $\theta$\label{sec:dynphase}}

For the purpose of calculating the QFI with respect to the dynamical phase $\theta$, and thereby the ultimate achievable precision for an estimate $\tilde{\theta}$ of $\theta$ \cite{PhysRevLett.72.3439}, the circuit-based model $V\ket{0}_{1}\ket{0}_{2}$ from (\ref{eqn:ccc}) is equivalent to the simpler state  $e^{i\theta a_{2}^{\dagger}a_{2}}\ket{\psi(g)}$ with $\ket{\psi(g)}=e^{g(a_{1}^{\dagger}a_{2}^{\dagger}-h.c.)}\ket{0}_{1}\ket{0}_{2}$,  because the last circuit element of (\ref{eqn:ccc}) is independent of $\theta$ so cannot change the QFI. Estimation of $\theta$ in the circuit-based description therefore falls into the shift model. A consequence of this fact is that the QFI is a $\theta$-independent constant proportional to the variance of the operator that couples to the parameter:
\begin{align}
\text{QFI}(\theta)&:=4\bra{\psi(g)} (a_{2}^{\dagger}a_{2})^{2} \ket{\psi(g)} - 4\bra{\psi(g)}a_{2}^{\dagger}a_{2}\ket{\psi(g)}^{2} \nonumber \\
&= 4N(g)(N(g)+1)\nonumber \\
&=E(g)^{2}-1
\label{eqn:thqficirc}
\end{align}
with $N(g):= \sinh^{2}g$ corresponding to half the number of expected photons of the system and $E(g)=2N(g)+1$ is the total energy of the system (including vacuum energy $1/2$ per mode). 

 The role of the last circuit element in (\ref{eqn:ccc}) is to allow one to obtain an estimate $\tilde{\theta}$ with error scaling as $E(g)^{-2}$ by measuring the total photon number $a_{1}^{\dagger}a_{1}+a_{2}^{\dagger}a_{2}$ and utilizing method of moments estimation. Specifically, 
it is a main result of \cite{PhysRevA.33.4033} that taking $O=a_{1}^{\dagger}a_{1}+a_{2}^{\dagger}a_{2}$ gives
\begin{equation}
\lim_{\theta\rightarrow 0}\text{SN}_{O}(\theta) = \sinh^{2}2g= 4N(g)(N(g)+1)
\label{eqn:snr}
\end{equation}
confirming that equality in the inequality $\sup_{O}\text{SN}_{O}(\theta) \le \text{QFI}(\theta)$ is saturated by taking $O$ to be the total photon number operator and $\theta\rightarrow 0$ on the left hand side. However, $\theta=0$ is not a valid operating point of the interferometer due to the fact that $V=\mathbb{I}$ at this point, which implies that the output modes are in the vacuum state. In other words, there is no interference at all and the sensitivity per photon diverges.

For the Hamiltonian model, we consider estimation of the dynamical phase $\theta$ at the optimal working point $\phi=\pi$ for estimation of the kinematical phase. The justification for this is that this submanifold is maximally distinguishable from vacuum, as can be seen from the effective nonlinearity $g\sin{\phi\over 2}$. Compared to estimation of the relative phase $\phi$, estimation of $\theta$ is more closely related to the estimation task considered in \cite{PhysRevA.33.4033} because $\theta$ couples the $\mathfrak{su}(1,1)$ generator $K_{3}$ in the Hamiltonian (\ref{eqn:hhhh}). Because 

Therefore, assuming a general $\theta >0$ in (\ref{eqn:hhhh}) and taking $\phi=\pi$, we proceed to analyze optimal estimation of $\theta$. At $\phi=\pi$, the covariance matrix for $\ket{\psi(g,\theta)}:=e^{-iH}\ket{0}_{1}\ket{0}_{2}$  is given by
\begin{align}
\Sigma_{\ket{\psi(g,\theta)}}&={1\over 2}\begin{pmatrix}A&B\\B&A\end{pmatrix}
\label{eqn:cov2}
\end{align}
where in \textbf{Domain 1} defined by $\theta^{2}\ge \lambda^{2}$ (introducing the variables $\lambda:=2g$ and $x=\sqrt{\theta^{2}-\lambda^{2}}$ for convenience)
\begin{align}
A&=  {\theta^{2}-\lambda^{2}\cos 2x\over x^{2}} \mathbb{I}_{2} \nonumber \\
B&= -{2\lambda \theta \over x^{2}}Z\sin^{2}x - {\lambda\over x}X\sin 2x
\label{eqn:rcov}
\end{align}
and in \textbf{Domain 2} defined by $\theta^{2}< \lambda^{2}$ (introducing the variables $\lambda:=2g$ and $x=\sqrt{\lambda^{2}-\theta^{2}}$ for convenience),
\begin{align}
A&=  {-\theta^{2}+\lambda^{2}\cosh 2x\over x^{2}} \mathbb{I}_{2} \nonumber \\
B&= -{2\lambda \theta \over x^{2}}Z\sinh^{2}x - {\lambda\over x}X\sinh 2x.
\label{eqn:rcov}
\end{align}
That the covariance matrices converge in the limit $\theta\rightarrow \lambda$ is readily verified. We use the following formula for the QFI with respect to $\theta$, which can be derived from, e.g., (\ref{eqn:hamqfi}) \cite{PhysRevA.85.010101}
\begin{equation}
\text{QFI}(\theta)={1\over 4}\text{tr}\left[ \left( \Sigma_{\ket{\psi(g,\theta)}}^{-1}\partial_{\theta}\Sigma_{\ket{\psi(g,\theta)}}\right)^{2}\right].
\end{equation}
As one can see from from Fig.\ref{fig:ttt},  the dependence of $\text{QFI}(\theta)$ on $\theta$ in Domain 1 is non-monotonic and in the present work we will only note the observation that although Fig. \ref{fig:ttt} suggests that Heisenberg scaling of $\text{QFI}(\theta)$ is not obtainable in Domain 1, the maximal QFI in Domain 1 scales as $E(g)^{\alpha}$ with $\alpha >1$, indicating that operating the $SU(1,1)$ interferometer in this parameter domain provides an advantage in estimation precision compared to coherent states of the same energy. Instead, we focus on Domain 2, and specifically the optimal operating point at $\theta = 0$. We find that $\text{QFI}(\theta)$ obeys the following scaling as $g\rightarrow \infty$, where the energy $E(g):={1\over 2}\text{tr}\Sigma_{\ket{\psi(g,0)}}$,
\begin{equation}
\text{QFI}(\theta)\big\vert_{\theta=0} \sim 4\left( {E(g)\over \ln 2E(g) }\right)^{2}.
\label{eqn:qfidyn}
\end{equation}
This scaling implies, through the quantum Cram\'{e}r-Rao bound, a logarithmically-modified Heisenberg scaling for the precision of an optimal unbiased estimator of $\theta$.

 \begin{figure}[t]
    \centering
    \includegraphics[scale=0.8]{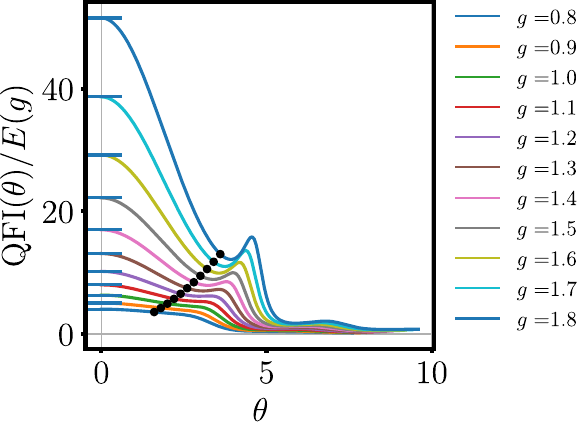}
    \caption{$\text{QFI}(\theta)$ relative to the total energy on the dynamical phase domain $\theta = [0,2g + 6]$  for various $g$. Black dots show the boundary of Domain 1 and Domain 2 at $\theta = 2g$, and blue horizontal lines show the relative QFI corresponding to the analytical result  $\text{QFI}(\theta=0)={\sinh^{4}2g\over g^{2}}$. Analogous data for (\ref{eqn:thqficirc}) are not shown because the $\text{QFI}(\theta)$ is independent of $\theta$ and $E(g)=1$ at $\theta=0$.}
    \label{fig:ttt}
\end{figure}

One of the main results of the analysis of the circuit model (\ref{eqn:ccc}) (or (\ref{eqn:ykyk}) in Ref.\cite{PhysRevA.33.4033}) is the $\theta = 0$ operating point that is characterized by Heisenberg scaling of the precision of an estimator $\tilde{\theta}$ obtained from a total photon number measurement (see (\ref{eqn:snr}) above). At the operating point $\theta=0$, there are two major distinctions to be made between the circuit models and the Hamiltonian model. The first is that for any value of $g$, the $\theta=0$ operating point defines a trivial quantum circuit, whereas the $\theta=0$ operating point of the Hamiltonian model does not define a trivial dynamics. The second distinction is that for the Hamiltonian model, the total photon number observable $O=a_{1}^{\dagger}a_{1}+a_{2}^{\dagger}a_{2}$ does not give a value of $\lim_{\theta\rightarrow 0}\text{SN}_{O}(\theta)$ that exhibits Heisenberg scaling, unlike the result in (\ref{eqn:snr}) above. A more surprising fact is that un-squeezing the total photon number operator $O$ to form an observable $O_{\text{sq}}$ such that $\langle O_{\text{sq}}\rangle_{e^{-iH}\ket{0}\ket{0}} = \langle O\rangle_{\ket{\xi(g,\theta)}}$ where $\ket{\xi(g,\theta)}:=e^{2ig(a_{1}^{\dagger}a_{2}^{\dagger}+h.c.)}e^{-iH}\ket{0}\ket{0}$ also does not result in a signal-to-noise ratio exhibiting Heisenberg scaling. This indicates that even the intuition behind using the total photon number readout in \cite{PhysRevA.33.4033} is not valid for the Hamiltonian model. 

An approach that yields a satisfactory readout is based on the following Lemma
\begin{lem}\label{lem:lll1} Let $\lambda^{2}>\theta^{2}$. Then, as $\lambda\rightarrow \infty$,
\begin{equation}
e^{-iH}\ket{0}\ket{0} - e^{\sqrt{\lambda^{2}-\theta^{2}}\left( e^{i(\pi + \cos^{-1}{\theta\over \lambda})} a_{1}^{\dagger}a_{2}^{\dagger} - h.c.\right)}\ket{0}\ket{0}
\label{eqn:hjhj}
\end{equation}
converges to zero in $\ell^{2}(\mathbb{C})^{\otimes 2}$, where the range  $[{\pi\over 2},\pi]$ is taken for $\cos^{-1}$.
\end{lem}
\begin{proof}
We first show that $e^{-iH}\ket{0}\ket{0}$ is actually a two mode squeezed state $e^{\zeta a_{1}^{\dagger}a_{2}^{\dagger} - \overline{\zeta}a_{1}a_{2}}\ket{0}\ket{0}$ with appropriate $\zeta\in\mathbb{C}$. Consider (\ref{eqn:cov2}), (\ref{eqn:rcov}) defining the covariance matrix of $e^{-iH}\ket{0}\ket{0}$ in Domain 2 at $\phi=\pi$ and define the parameters
\begin{align}
f&:= {1\over 2}\cosh^{-1}\left( { - \theta^{2}+\lambda^{2}\cosh 2\sqrt{\lambda^{2} - \theta^{2}}  \over \lambda^{2} - \theta^{2}} \right) \nonumber \\
\xi&:= \pi +\cos^{-1}\left( \theta \sin\sqrt{\lambda^{2} - \theta^{2}} \over \sqrt{- \theta^{2}+\lambda^{2}\cosh^{2}\sqrt{\lambda^{2} - \theta^{2}} }  \right)
\end{align}
where $f\rightarrow {1\over 2}\cosh^{-1}(1+2\lambda^{2})$ and $\xi\rightarrow \pi$ as $\lambda \rightarrow \theta^{+}$ and $f\rightarrow \sqrt{\lambda^{2}-\theta^{2}}$ and $\xi\rightarrow \pi + \cos^{-1}{\theta\over\lambda}$ as $\lambda \rightarrow \infty$. Taking $\zeta:= fe^{i\xi}$, one finds that the covariance matrices of $e^{-iH}\ket{0}\ket{0}$ and $e^{\zeta a_{1}^{\dagger}a_{2}^{\dagger} - \overline{\zeta}a_{1}a_{2}}\ket{0}\ket{0}$ agree. Since these states also have equal displacement in phase space (namely, zero displacement), they are equal as states. The $\lambda \rightarrow \infty$ asymptotics of $f$ and $\xi$ give the element of $\mathfrak{su}(1,1)$ in the exponent of (\ref{eqn:hjhj}).
\end{proof}
As a function of $\theta$, the large $\lambda$ approximation to the Hamiltonian model obtained in this Lemma is similar in structure to the Hamiltonian model as a function of the relative phase $\phi$ at $\theta = 0$ in (\ref{eqn:thzero}). In fact, taking $\sin {\phi\over 2} = \sqrt{1-{\theta^{2}\over \lambda^{2}}}$ in (\ref{eqn:thzero}) then mapping $\phi \mapsto -\phi$, one finds due to Lemma \ref{lem:lll1} that $e^{-iH}\ket{0}\ket{0}$ converges to (\ref{eqn:thzero}). This fact can be interpreted as a duality between the relative phase and the dynamic phase.  In Section \ref{sec:relphase} we found that measurement of the observable $O$ in (\ref{eqn:oop}) is optimal for estimation of relative phase $\phi$ near the optimal operating point $\phi = \pi$, which corresponds to the $\lambda \rightarrow \infty$ limit of Domain 2. Calculating the signal-to-noise ratio for readout $O$ in the large $\lambda$ approximation to $e^{-iH}\ket{0}\ket{0}$ from Lemma \ref{lem:lll1}, one obtains $\lim_{\theta\rightarrow 0}\text{SN}_{O}(\theta)  = {\sinh^{4}2g \over g^{2}}$, saturating the QFI. This establishes (\ref{eqn:oop}) as an asymptotically optimal readout for the dynamical phase for $\lambda^{2} \gg \theta^{2}$ and $\theta \rightarrow 0$.

\section{Discussion}
In our comparison of a circuit-based model and Hamiltonian model for the $SU(1,1)$ interferometer, we have fixed the parameter $\theta$ as the symbol that couples to the $\mathfrak{su}(1,1)$ generator $K_{3}$. However, unlike the Hamiltonian model, we do not see that the circuit-based model contains any meaningful differentiation between a dynamical phase corresponding to interaction time and the relative phase that corresponds to a kinematical phase difference between the modes. For instance, one may argue that the appropriate phase coupling to $K_{3}$ in the circuit model is some non-dynamical relative phase $\phi$, in which case it is clear that at $\phi=0$ both the circuit-based model and the Hamiltonian based model leave the vacuum invariant. If this interpretation is assumed, then the results of Section \ref{sec:relphase} imply different optimal operating points for the two models ($\phi=0$ for the circuit model and $\phi=\pi$ for the Hamiltonian model), with nontrivial dynamics occurring only for the latter model. Alternatively, one could replace the generators of the optically active circuit elements in (\ref{eqn:su11ccc}) and (\ref{eqn:ykyk}) by rotated generators in the span of $K_{1}, K_{2}\in \mathfrak{su}(1,1)$ and refer to the rotation phase as the analogue of $\phi$, although the problem of justifying the description by a Hamiltonian quadratic in the canonical operators remains. By contrast, the formulation of the Hamiltonian model and the different scalings of the QFI in (\ref{eqn:thqficirc}) and (\ref{eqn:qfidyn}) show that the phases $\phi$ and $\theta$ are physically distinct, arising from the overlap of optical modes and the interaction time, respectively. Recent proposals for embedded passive interferometers \cite{PhysRevLett.128.033601,agarwal2025}

It is an intriguing fact that despite the suggestive similarity between circuit-based optical interferometers and the circuit description of quantum computation, recent proposals for optical quantum computation utilize the measurement-based setting for quantum computation instead of the circuit-based setting. The reason is that CV circuit elements corresponding to optically active dynamics are noisy and difficult to implement as in-line operations, so a better strategy is to prepare a non-classical resource state (e.g., a cluster state) at the outset and manipulate its stored quantum information. The Hamiltonian description of the $SU(1,1)$ interferometer can be generalized to describe non-classical Gaussian CV state generation by aligning multiple downconversion processes, which could inform new approaches for generating CV cluster states.

Lastly, we note that we have not analyzed estimation of the nonlinearity $g$ in the Hamiltonian model because it should be considered simultaneously with the problem of estimation of the renormalized nonlinearity $g\sin{\phi\over 2}$, which is a (nonlinear) multiparameter sensing problem. The full task of optimal $\mathfrak{su}(1,1)$ Hamiltonian estimation is the subject of intensive research, similar to its $\mathfrak{su}(2)$ counterpart \cite{ogv,du2025characterizingresourcesmultiparameterestimation}.

\onecolumngrid

\begin{acknowledgements} 
The author thanks Alberto Marino for insights into the experimental operation of an $SU(1,1)$ interferometer and Carlton Caves for technical discussions. This work was supported by the National Quantum Information Science Research Centers, the Quantum Science Center, the ASC Beyond Moore’s Law project, and the Laboratory Directed Research and Development program of Los Alamos National Laboratory.

\end{acknowledgements}

\bibliography{refs}

\appendix

%\onecolumngrid

\newpage

\section{Phase space properties of the Hamiltonian model\label{sec:aaa}}

In this Appendix, we derive the covariance matrix for the full state $\ket{\psi(g,\theta,\phi)}:= e^{-iH}\ket{0}\ket{0}$ and its symplectic eigenvalues, with the Hamiltonian $H$ defined 
\begin{align}
H&=H_{0} + H_{\text{int}} \nonumber \\
H_{0}&= \theta(a_{1}^{\dagger}a_{1}+a_{2}^{\dagger}a_{2})\nonumber \\
H_{\text{int}}&= {2g\sin{\phi\over 2}}\left( ie^{-i{\phi\over 2}}a_{1}^{\dagger}a_{2}^{\dagger} +h.c.\right).
\label{eqn:a1}
\end{align}
Note that we have referred to unitary operators such as $e^{-iH}$ or the circuit (\ref{eqn:su11ccc}) simply as elements of $SU(1,1)$. Since the input state under consideration is always $\ket{0}\ket{0}$, the well-definedness of these operators on $\ell^{2}(\mathbb{C})^{\otimes 2}$ up to a phase follows the construction of the metaplectic representation of $SU(1,1)$, i.e., the simply connected double cover of $SU(1,1)$, using Nelson's theorem (Proposition 4.49 of \cite{folland})).

The covariance matrix is defined by
\begin{align}
 \Sigma_{\ket{\psi(g,\theta,\phi)}}&= \langle \psi(g,\theta,\phi)\vert R^{T}\circ R \vert \psi(g,\theta,\phi)\rangle,
 \end{align}
where $R:=(q_{1},p_{1},q_{2},p_{2})$ a row vector of canonical observables,  $\circ$ the entrywise Jordan product $A\circ B:= {1\over 2}(AB+BA)$, and the symplectic form on the phase space $\mathbb{R}^{4}$ is $\Delta = \bigoplus_{j=1}^{2}\begin{pmatrix} 0&1\\-1&0\end{pmatrix}$.   Corresponding to the $SU(1,1)$ element $e^{-iH}$ in the two-mode bosonic realization, there is a symplectic matrix $T_{e^{-iH}} \in Sp(4,\mathbb{R})$ is obtained from the equation (the metaplectic representation of $Sp(4,\mathbb{R})$)
\begin{equation}
    e^{iH}Re^{-iH} = RT_{e^{-iH}}
\end{equation}
and satisfies $T_{e^{-iH}}^{\intercal}\Delta T_{e^{-iH}}=\Delta$ by definition \cite{serafini}.  Introducing the variables $\lambda:=2g \sin {\phi\over 2}$ and $x:=\sqrt{\theta^{2}-\lambda^{2}}$ and $\xi = {\pi \over 2} - {\phi\over 2}$ in the \textbf{Domain 1} defined by $4g^{2}\sin^{2}{\phi\over 2} \le \theta^{2}$, one obtains
\begin{equation}
    T_{e^{-iH}}=\begin{pmatrix}
        \cos x & - {\theta\sin x\over x} & {\lambda \sin \xi \sin x\over x} & -{\lambda \cos \xi \sin x\over x} \\
        {\theta\sin x\over x} & \cos x & -{\lambda \cos \xi \sin x\over x} & -{\lambda \sin \xi \sin x\over x}\\
        {\lambda \sin \xi \sin x\over x}& -{\lambda \cos \xi \sin x\over x}& \cos x & -{\theta\sin x\over x} \\
        -{\lambda \cos \xi \sin x\over x}& -{\lambda \sin \xi \sin x\over x} & {\theta\sin x\over x} & \cos x
    \end{pmatrix}
\end{equation}
The covariance matrix is obtained as 
\begin{equation}
    \Sigma_{\ket{\psi(g,\theta,\phi)}} = {1\over 2}T_{e^{-iH}}^{\intercal}T_{e^{-iH}}
    \label{eqn:d1cov}
\end{equation}
and the total energy, including vacuum energy $1/2$ per mode, is given by 
\begin{equation}
   {1\over 2}\text{tr} \Sigma_{\ket{\psi(g,\theta,\phi)}}= {\theta^{2} - \lambda^{2}\cos 2x \over x^{2}}
   \label{eqn:e1}
\end{equation}
In \textbf{Domain 2} defined by $4g^{2}\sin^{2}{\phi\over 2} \ge \theta^{2}$, we introduce the variables $\lambda:=2g \sin {\phi\over 2}$ and $x:=\sqrt{\lambda^{2}-\theta^{2}}$ and $\xi = {\pi \over 2} - {\phi\over 2}$. The symplectic matrix is then given by
\begin{equation}
    T_{e^{-iH}}=\begin{pmatrix}
        \cosh x & - {\theta\sinh x\over x} & {\lambda \sin \xi \sinh x\over x} & -{\lambda \cos \xi \sinh x\over x} \\
        {\theta\sinh x\over x} & \cosh x & -{\lambda \cos \xi \sinh x\over x} & -{\lambda \sin \xi \sinh x\over x}\\
        {\lambda \sin \xi \sinh x\over x}& -{\lambda \cos \xi \sinh x\over x}& \cosh x & -{\theta\sinh x\over x} \\
        -{\lambda \cos \xi \sinh x\over x}& -{\lambda \sin \xi \sinh x\over x} & {\theta\sinh x\over x} & \cosh x
    \end{pmatrix}
\end{equation}
with covariance matrix having the same expression as (\ref{eqn:d1cov}) and total energy
\begin{equation}
   {1\over 2}\text{tr} \Sigma_{\ket{\psi(g,\theta,\phi)}}= { \lambda^{2}\cosh 2x - \theta^{2} \over x^{2}}.
   \label{eqn:e2}
\end{equation}

The reduced state $\rho_{1}:= \text{tr}_{2}\psi(g,\theta,\phi)$ ($\rho_{2}$ is the same due to the symmetry of the modes) has symplectic eigenvalue equal to $1/2$ times the expressions in (\ref{eqn:e1}), (\ref{eqn:e2}) in the respective parameter domains.
This fact allows to obtain entanglement entropy, relative entropies \cite{10.1063/1.5007167}, and fidelities \cite{PhysRevA.86.022340} of the involved modes, and is most easily proved using the expression of $e^{-iH}\ket{0}\ket{0}$ as a two mode squeezed state in Lemma \ref{lem:lll1} (similarly for \textbf{Domain 1}). Specifically, the covariance matrix for the state $e^{\zeta a_{1}^{\dagger}a_{2}^{\dagger} - \overline{\zeta}a_{1}a_{2}}\ket{0}\ket{0}$ with $\zeta : = fe^{iu}$ is
\begin{equation}
    {1\over 2}\cosh 2f \,\mathbb{I}_{2}\otimes \mathbb{I}_{2} + {1\over 2}\sinh 2f \cos u \, X\otimes Z + {1\over 2}\sinh f \sin u \, X\otimes X
\end{equation}
with
\begin{align}
    \cosh 2f &= {\theta^{2}-\lambda^{2}\cos 2x \over x^{2}} \nonumber \\
    \cos u &= {x\cos x \sin \xi - \theta\sin x \cos \xi \over \sqrt{\theta^{2}-\lambda^{2}\cos^{2}x }}
\end{align}
in \textbf{Domain 1} and
\begin{align}
    \cosh 2f &= {\lambda^{2}\cosh 2x - \theta^{2}\over x^{2}} \nonumber \\
    \cos u &= {x\cosh x \sin \xi - \theta\sinh x \cos \xi \over \sqrt{\lambda^{2}\cosh^{2}x - \theta^{2}}}
\end{align}
in \textbf{Domain 2}.

\end{document}